\documentclass[runningheads]{llncs}
\usepackage{graphicx}
\usepackage{appendix}
\usepackage{amsfonts}
\usepackage{amsmath,amssymb}
\usepackage{mathrsfs}
\usepackage{amsfonts}
\usepackage{algorithm}  
\usepackage{algorithmicx}  
\usepackage{algpseudocode}  
\usepackage{float}
\begin{document}
\title{Environment Modeling During Model Checking of Cyber-Physical Systems}
%
%
\author{Guangyao Chen\inst{1}\orcidID{0000-0003-1367-8331} \and Zhihao Jiang\inst{2}\orcidID{0000-0002-6730-6915}}
\authorrunning{G. Chen \& Z. Jiang}
%
\institute{ShanghaiTech University, Shanghai, China \\
\email{chengy2@shanghaitech.edu.cn}\and
ShanghaiTech University, Shanghai, China\\
\email{jiangzhh@shanghaitech.edu.cn}}

\maketitle              
\begin{abstract}
Ensuring the safety and efficacy of Cyber-Physical Systems (CPSs) is challenging due to the large variability of their operating environment.
Model checking has been proposed for validation of CPSs, but the models of the environment are either too specific to capture the variability of the environment, or too abstract to provide counter-examples interpretable by experts in the application domain.
Domain-specific solutions to this problem require expertise in both formal methods and the application domain, which prevents effective application of model checking in CPSs validation.
A domain-independent framework based on timed-automata is proposed for abstraction and refinement of environment models during model checking.
The framework maintains an abstraction tree of environment models, which provides interpretable counter-examples while ensuring coverage of environment behaviors.
With the framework, experts in the application domain can effectively use model checking without expertise in formal methods.
\end{abstract}
\section{The Emergence of Cyber-Physical Systems (CPS)}
With the development of technologies, it is now possible to develop software that can make real-time decisions under complex situations.
As a result, problems in the physical world can be solved by software-controlled physical systems with little to none human intervention.
These \emph{Cyber-Physical Systems (CPSs)} are relieving human from tedious jobs and dangerous working environment, and have improved quality of lives and the overall efficiency of the society.

\begin{figure}
\centering
	\includegraphics[width=0.7\textwidth]{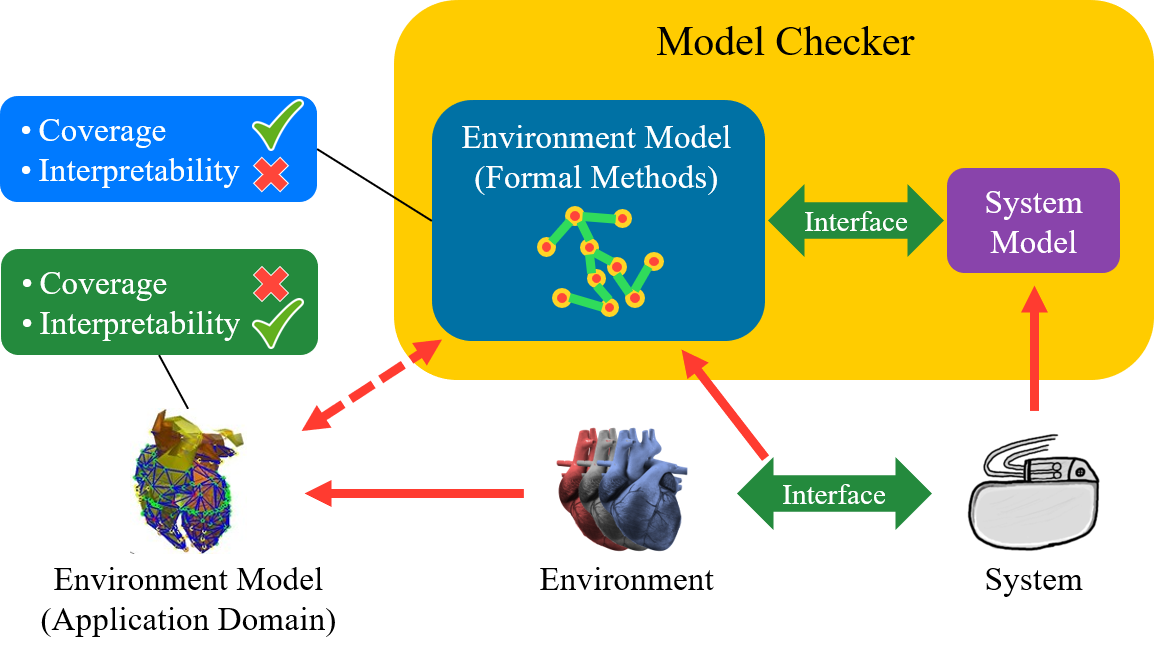}
	\caption{Environment models are developed from two different perspectives, but neither satisfy the need for validating CPSs. A domain-independent framework that can balance the coverage and interpretability of environment models is needed for effective use of model checking in CPS domains.}
	\label{fig:environment}
\end{figure}

With human-operated systems, decisions are made by domain experts, who have been trained to deal with the complexity and variability of the environment, and are responsible for preventing safety hazards.
For CPSs with increasing autonomy, malfunctions of the systems cannot receive timely human intervention, which can cause serious harm to people and properties in its operating environment, especially in safety-critical domains like medical devices \cite{Mederror}.
Manufacturers of CPSs are required to demonstrate the safety and efficacy of the systems, especially their software components.

\subsection{Validation of CPS Using Model Checking}
With increasing autonomy, CPS are required to make correct decisions under ALL possible environment conditions.
CPSs cannot be exhaustively tested as the amount of environment conditions is infinite.
Model checking exhaustively examines the reachable states of a model, which is suitable for validation of CPSs \cite{cps1,cps2}.
By modeling the CPS and its operating environment, model checking tools can prove that the CPS satisfies safety and efficacy requirements under conditions specified in the environment models, or provide counter-examples when requirements are violated. 
Development cost can be significantly reduced when bugs are found and safety guarantees are provided in the early development stage.

\subsection{Environment Modeling for CPS}
The environment models represent assumed environment conditions, and the results of the model checking can only support safety and efficacy of CPS under these conditions.
Model(s) of the environment should satisfy the following requirements:\\
\textbf{Coverage: }
Environment model(s) of CPSs should either 1) cover all environment conditions, or 2) cover all environment behaviors observable to the system.
These two conditions are equivalent but the second one can be better defined and quantified.\\
\textbf{Interpretability: }The safety and efficacy of CPS are evaluated on the states of the environment.
i.e. The patient's condition should be improved with a medical device compared to without the device.
In order to judge whether environment conditions have been improved, the models of the environment should have states and executions that are interpretable by experts in the application domain.

Unfortunately, these two requirements conflict with each other in most cases, and no single model can satisfy both.
Relaxing domain-specific constraints within the models may introduce new observable behaviors, which increases coverage at the cost of interpretability.
As shown in Fig. \ref{fig:environment}, on one hand, experts in the application domain develop models to study the mechanisms of the problem.
These models have great interpretability, but are not suitable for model checking due to their inadequate coverage.
On the other hand, abstract models of the environment are created in the formal methods community to cover observable behaviors of the environment.
These models are usually abstracted from the interface between the system and the environment, and coverage can be easily quantified.

Formal relationships between the formal models and the domain models are needed in order to balance coverage and interpretability.
However, establishing connections require expertise in both formal methods and the application domain.

In \cite{STTT13}, Jiang et. al proposed the use of over-approximation \cite{Predicate} to increase the coverage of environment models in closed-loop model checking of implantable cardiac devices, and refine the environment models to provide interpretability to counter-examples.
Unfortunately, the proposed method requires abstraction rules based on extensive domain knowledge, which cannot be applied directly in other domains.
A domain-independent framework for environment modeling is essential for model checking to be effectively adopted.
\begin{figure}[t]
\centering
	\includegraphics[width=0.45\textwidth]{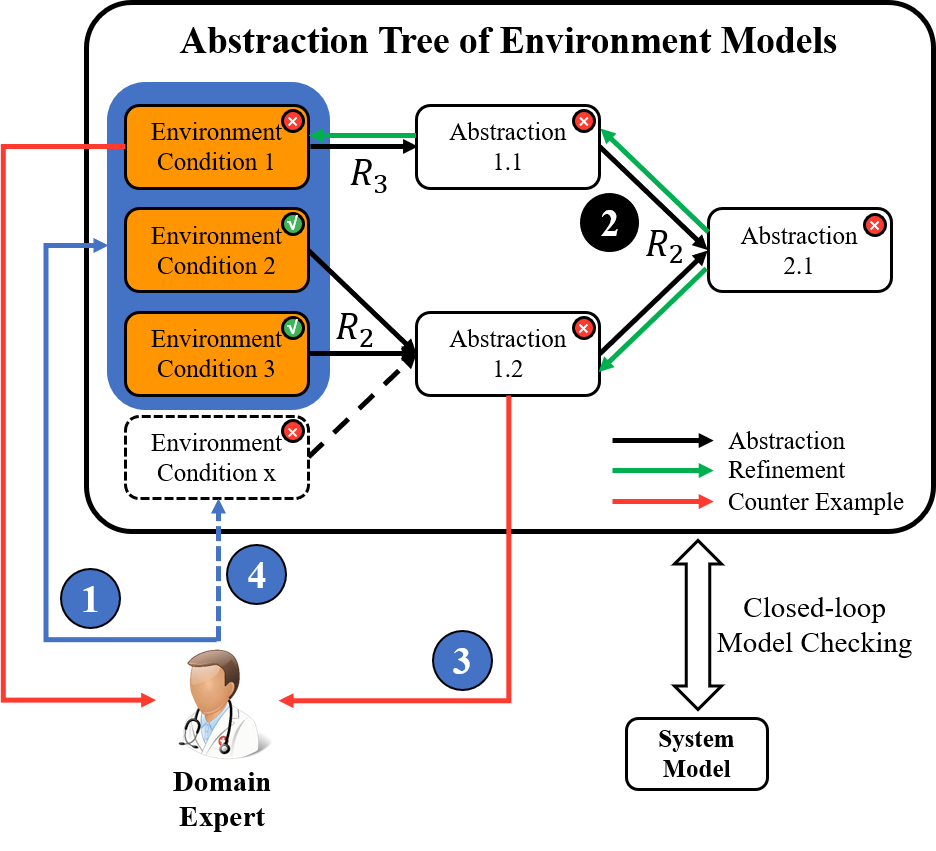}
	\caption{Experts in the application domain provide a set of base environment models (1), and the framework returns a set of counter-examples with the most refined context (3). The experts can also provide additional base models after analyzing the returned counter-examples and their corresponding levels of abstraction (4). Inside the framework an abstraction tree is created by abstracting the set of base models using domain-independent abstraction rules (2), which is hidden from the domain expert.}
	\label{fig:framework}
\end{figure}

\section{Domain-independent Model Checking Framework with Environment Abstraction \& Refinement}
In this project, we propose a model checking framework with environment abstraction \& refinement which can balance the coverage and interpretability of environment models during model checking.
The framework is also domain-independent such that the abstraction and refinement of environment models do not require domain-specific information.

The framework is illustrated in Fig. \ref{fig:framework}, which involves four main steps:\\
\textbf{Step 1: Initial Set of Environment Models:}\\
A set of base environment models are first provided by domain experts, which represent prior knowledge of environment conditions that the CPS may encounter.
These base environment models does not provide adequate coverage, 
but their execution traces, including counter-examples returned from the model checker, are interpretable by domain experts.\\
\textbf{Step 2: Construction of the abstraction tree:} \\
Domain-independent abstraction rules are then applied to the base environment models so that the abstract model over-approximates the original model(s), covering more observable behaviors of the environment.
By abstracting and combining models of the environment, an \textit{abstraction tree} of environment models can be built. 
This step is hidden from the experts in the application domain, so expertise in formal methods is not required for using the framework.\\
\textbf{Step 3: Model Checking and Counter-example Refinement:} \\
The safety and efficacy of the system model can then be validated using the abstraction tree of environment models.
The system model is first verified against the root environment model.
If the requirements are satisfied, the system model is safe under all possible environment conditions.
Otherwise, the system model is then verified against the environment model(s) that are children of the root environment model.
The process traverse the abstraction tree in the Breath-First Search (BFS) manner until 1) the leaves of the abstraction tree is reached, or 2) all children of the current environment model satisfy the requirement.
The counter-example(s) are attached to the abstraction tree, and returned to the domain experts for further analysis.\\
\textbf{Step 4: Environment Model Refinement:}\\
Depending on the "completeness" of the set of base models and the topology of the abstraction tree, the refined counter-examples returned may not correspond to the leaf nodes in the abstraction tree.
In this case the violations of requirement happen in environment conditions that are not included in the set of base models.
Moreover, the counter-examples returned may not have adequate context for comprehensive interpretation.
Domain experts can create new base models, which are refinements of the model that returned the counter-examples.
i.e. in Fig. \ref{fig:framework}, environment condition x can be created by "subtracting" environment condition 2 and 3 from Abstraction 1.2.


The framework hides domain knowledge in formal methods from experts in the application domain, so that model checking becomes a more friendly tool for validating CPSs.

\section{Abstraction Tree Construction with Timed Automata}
In order for the framework to achieve domain-independence, the application and selection of the abstraction rules should not contain knowledge in the application domain.
In this project, we use timed automata \cite{timed_automata} as modeling formalism and UPPAAL \cite{uppaal} as model checker.
Abstraction rules targeting the structure of timed-automata are proposed and their effect on observable environment behaviors are informally proved.
Formal proofs of the Theorems can be found in \cite{}.

    

\subsection{Timed Automata and Model Checker UPPAAL}
Timed-automata \cite{timed_automata} is a formalism developed to model real-time systems.
It has the expressiveness for modeling complex system behaviors \cite{cps3}, and the simplicity for decidable reachability. 
Timed automata also supports non-determinism, which can be used to capture the uncertainty within the environment.
The framework proposed in this project is applicable when both the system and the environment are modeled using timed automata.

A timed automaton is a tuple $(L,l_0,X,A,E,G,I)$, where
\begin{enumerate}
    \item $L$ is a set of locations.
    \item $l_0\in L$ is the initial location.
    \item $X$ is the set of clocks.
    \item $A$ is a set of actions, including sending actions (a!) and receiving actions (a?).
    \item $E\subseteq L\times A\times 2^X \times L$. \\
    An edge (transition) $e\in E$ is a tuple $(l,a,r,l')$, where $l$ is the start location, $a$ is the action, $r$ is the set of clocks to be reset and $l'$ is the target location.
    \item $G: E\times 2^X \times 2^\mathbb{N} \rightarrow \Psi_G$ assigns guards to edges.\\
    $G$ can be written as $G(E,X,N)=\{ g_i(e_i,X_i,N_i)\;|\;i \in \mathbb{N}$ and $i\le len(E) \}$. Each $g_i$ denotes the guard of the edge $e_i$, which constrains the set of clocks in $X_i$ with the set of lower bounds $N_i$.
    \item $I: L\times 2^X \times 2^\mathbb{N} \rightarrow \Psi_I$ assigns invariants to locations.\\
    $I$ can be written as $I(L,X,M) = \{inv_i(l_i,X_i,M)\;|\;i \in \mathbb{N}$ and $i\le len(E) \}$. Each $inv_i$ denotes the invariant of the location $l_i$, which constrains the set of clocks in $X_i$ with the set of upper bounds $N_i$.
    \item \(\Psi\) is the clock constraints for clock variables \(X\).\\
    \(\Psi:=x \perp n \;\|\; \Psi_{1} \wedge \Psi_{2}\), where \(x\in X\),
    \(\perp \in\{\leq,\;\geq\},\) and \(n \in \mathbb{N}\).\\
    For particular guard and invariant clock constraints, we have
    \begin{itemize}
      \item
        \(\Psi_G \in\Psi^X\) and
        \(\Psi_G:=x \ge n \;\|\; \Psi_{1} \wedge \Psi_{2}\)
      \item
        \(\Psi_I \in\Psi^X\) and
        \(\Psi_I:=x \le n \;\|\; \Psi_{1} \wedge \Psi_{2}\)
    \end{itemize}
\end{enumerate}


Multiple timed automata can run in parallel and interact with each other via actions.
i.e. the system model and the environment model form a closed-loop system.
We use $\mathscr{A}_1|\mathscr{A}_2$ to represent automata composition. The semantics ~\cite{UPPAAL_Tutorial} is defined as a labelled transition system
\(<S, s_0, \rightarrow>\), where
\begin{enumerate}
\item
  \(S\subseteq L\times R^C\) is the set of states,
\item
  \(s_0 = <l_0, u_0>\) is the initial state,
\item
  \(u: C \rightarrow \mathbb{R}_{\geq 0}\) is the function of a clock valuation and
\item
  \(\rightarrow \subseteq S\times (R_{\geq 0} \cup A)\times S\) is the transition relation such that:
  \begin{itemize}
  \item
    \((l, u) \stackrel{d}{\rightarrow}(l, u+d)\) if \(\forall d^{\prime}: 0 \leq d^{\prime} \leq d \Longrightarrow u+d^{\prime} \in I(l,x,n)\), where \(x\in 2^X\) and \(n\in 2^\mathbb{N}\)
  \item
    \((l, u) \stackrel{a}{\rightarrow}\left(l^{\prime}, u^{\prime}\right)\)
    if there exists \(e=\left(l, a, r, l^{\prime}\right) \in E\) s.t. \(u \in G(e,x',n')\), \(u^{\prime}=[r \mapsto 0] u,\) and \(u^{\prime} \in I\left(l^{\prime},n'\right)\), where \(x'\in 2^X\) and \(n'\in \mathbb{N}\)
  \end{itemize}
\end{enumerate}

UPPAAL \cite{uppaal} is a model checking tool using timed automata as formalism, and is very friendly to people with little programming experience.
Users can model their system and its environment in a graphic interface, and counter-examples returned by the model checker are visualized in the simulator.\\
Fig. \ref{fig:TA_example} shows the composed timed automaton $Speaker|Translator$ that interact with each other via action $a1$ in UPPAAL.
The sending action $a1!$ in $Speaker$ is confined by guard $t>=5$ and invariant $t<=10$, which represents the uncertainty in behaviors $a1!$.

\begin{figure}[t]
\centering
	\includegraphics[width=0.35\textwidth]{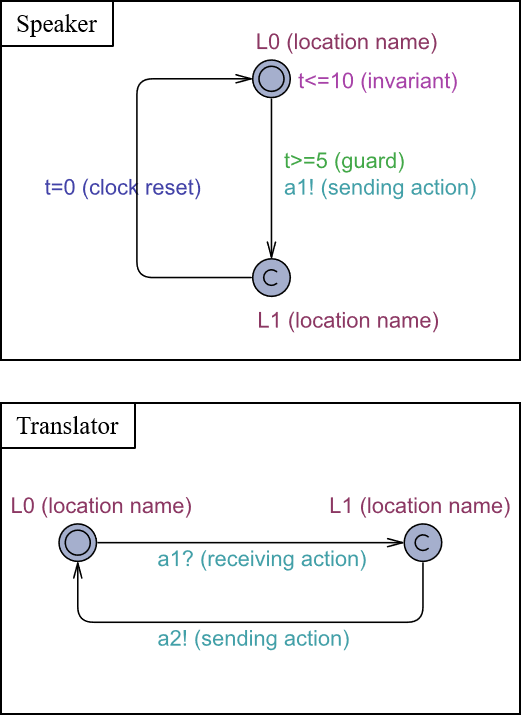}
	\caption{The composed model $Speaker|Translator$. The time interval between two consecutive action $a1$ is larger than 5 unit time and shorter than 10 unit time. The $Translator$ sends $a2$ immediately after receiving $a1$.}
	\label{fig:TA_example}
\end{figure}

\subsection{Prerequisite of the Framework}
Currently the framework is suitable for problems with the following constraints:
\begin{enumerate}
    
    \item The environment contains multiple independent agents interacting with each other via events.
    \item The system also interacts with the environment via events. 
    Only a subset of events in the environment are observable to the system, and the system operates base on the timing and patterns of these events.
    \item Both the system and the environment are modeled using timed automata.
    \item At a particular state of the environment, the observable events can occur within a timing interval $[T\_min, T\_max]$
    \item The differences among base environment models are parameters-only.
\end{enumerate}
\subsection{Coverage of Environment Behaviors}

The system can only observe a subset of actions $A_o \subseteq A$ in the environment.
Environment behavior is defined as \textit{timed word} \cite{timed_automata} over observable actions $A_o \subseteq A$, which is a pair 
$(\Sigma,T)$ 
where $\Sigma = \sigma_1 \sigma_2 ...$, $\sigma_i\in A_o$ represents the sequence of actions, and $T=\tau_1\tau_2 ...$, $\tau_i\in \mathbb{R}$ represents the global time the actions happened.
The timed language of a timed automaton $\mathscr{A}$ is the set of all the possible timed words of $\mathscr{A}$, which is represented as $\mathscr{L}(\mathscr{A})$.
The coverage of environment behaviors is then measured on the "size" of the language.\\
Followings are the formal definition of timed sequence, timed word and timed language.\\
\textbf{Definition: Timed sequence}\\
A timed sequence~\cite{TheoryTA} $\tau = \tau_1 \tau_2 ...$ is an infinite sequence of time values $\tau_i \in R$ with $\tau_i > 0$, satisfying the following constraints:
\begin{enumerate}
    \item 
    Monotonicity
    \begin{itemize}
        \item $\tau$ increases strictly monotonically, i.e. for all $i \geq 1$, we have $\tau_{i+1} > \tau_i$;
    \end{itemize}
    \item
    Progress
    \begin{itemize}
        \item For every $t \in R$, there is some $i\geq 1$ such that $\tau_i > t$.
    \end{itemize}
\end{enumerate}
\textbf{Definition: Timed word}\\
A timed word~\cite{TheoryTA} over observable actions $A_o \subseteq A$ is a pair $(\sigma, \tau)$ where $\sigma = \sigma_1 \sigma_2 ...$ were each $\sigma_i$ indicates whether an observable action is observed.\\
For example, let $A=\{a_{o1},a_{o2},a_{u1},a_{o3},a_{u2}\}$ is the set of all actions, $A_o=\{a_{o1},a_{o2},a_{o3}\}$ is the set of observable actions, and $A_{u}=\{a_{u1},a_{u2}\}$ is the set of unobservable actions.\\
If at time $\tau_1$, no observable actions is observed, then $\sigma_1=\langle 0,0,0 \rangle$.\\
If at time $\tau_2$, $a_{o1}$ and $a_{o3}$ are observed, then $\sigma_2=\langle 1,0,1 \rangle$.\\
\textbf{Definition: Timed language }\\
For a timed automaton $\mathscr{A}=(L,l_0,X,A,E,G,I)$, where $A$ is the set of actions. The timed language~\cite{TheoryTA} of $\mathscr{A}$ is the set of all the possible timed words of $\mathscr{A}$.

\subsection{Domain-independent Abstraction Rules}
A set of abstraction rules on timed automata is defined that can increase the coverage of observable behaviors of the environment.
The correctness of the abstraction rules are informally proved to provide intuition for the audience of this paper.
Interested audience can find formal proofs in the appendix.


\subsubsection{$\mathcal{R}_1$: Increase Transition Uncertainty$\;\;\;$}
\textbf{Intuition: } A transition $e=(l,a,r,l')$ is \textit{enabled} when the guard assigned to the transition $g(e,X_g,N)$ evaluates to true, and it has to be taken when the invariant of its source location $inv(l,X_i,M)$ is about to be violated due to the increase of $X_i$. 
The interval $[N,M]$ represents the uncertainty when event $a$ can occur.
If the interval is expanded, intuitively the constraints on sending the event $a$ are relaxed, and the new model covers more behaviors.\\
\textbf{Prerequisite: }None.\\
\textbf{Rule:} Given a timed automaton
$\mathscr{A}_1=(L,l_0,X,A,E,G,I)$, and two 
non-negative vector $\Delta_G=\{\delta^G_i\;|\;i\in \mathbb{N}$ and $i\le len(G)\}$ and $\Delta_I=\{\delta^I_i\;|\;i\in \mathbb{N}$ and $i\le len(I)\}$,
create another timed automaton 
$\mathscr{A}_2=\mathcal{R}_1(\mathscr{A},\Delta_G,\Delta_I)=(L,l_0,X,A,E,G^R,I^R)$ s.t.
\begin{itemize}
    \item $\forall g_i(e_i,X_i,N^1_i)\in G$, 
 $N^R_i=N^1_i-\delta^G_i$ for all $g^R_i(e_i,X_i,N^R_i)\in G^R$
 \item $\forall inv_i(l_i,X_i,M^1_i)\in I$,
 $M^R_i=M^1_i-\delta^I_i$ for all $inv^R_i(l_i,X_i,M^R_i)\in I^R$.
\end{itemize}


\begin{theorem}
$\mathscr{L}(\mathscr{A})\subseteq \mathscr{L}(\mathcal{R}_1(\mathscr{A},\Delta_G,\Delta_I))$ for non-negative $\Delta_G,\Delta_I$.
\end{theorem}
\textbf{Informal Proof: }First we prove that $\mathcal{R}_1(\mathscr{A},\Delta_G,\Delta_I)$ is a \textit{timed-simulation} of $\mathscr{A}$, which further implies that $\mathscr{L}(\mathscr{A})\subseteq \mathscr{L}(\mathcal{R}_1(\mathscr{A},\Delta_G,\Delta_I))$.
True subset can then be proved by construction, as there always exists timed words in $\mathcal{R}_1(\mathscr{A},\Delta_G,\Delta_I)$ which are not in $\mathscr{A}$.

\subsubsection{$\mathcal{R}_2$: Merge Models with the Same Structure}
\textbf{Intuition: } When the differences between two timed automata are confined to $N$ in guards $G(E,X,N)$ and $M$ in invariants $I(L,X,M)$, creating a timed automaton with the minimum of $N$ and the maximum of $M$ covers the behaviors of both models, as well as additional behaviors.
\textbf{Prerequisite: } The differences between $\mathscr{A}_1$ and $\mathscr{A}_2$ should be confined to the $N$ of guards $g(e,X,N)\in G$ and the $M$ of invariants $inv(l,X,M)$.\\
\textbf{Rule:} Given two timed automata
$\mathscr{A}_1=(L,l_0,X,A,E,G^1,I^1)$ and\\
$\mathscr{A}_2=(L,l_0,X,A,E,G^2,I^2)$,
create another timed automaton
$\mathscr{A}_3=\mathcal{R}_2(\mathscr{A}_1,\mathscr{A}_2)=(L,l_0,X,A,E,G^3,I^3)$ such that 
\begin{itemize}
    \item $\forall g^1_i(e_i,X_i,N^1_i)\in G^1$ and $\forall g^2_i(e_i,X_i,N^2_i)\in G^2$, $N^3_i=elm\_min(N^1_i,N^2_i)$ for all $g^3_i(e_i,X_i,N^3_i)\in G^3$
    \item $\forall inv^1_i(l_i,X_i,M^1_i)\in I^1$ and $\forall inv^2_i(l_i,X_i,M^2_i)\in I^2$, $M^3_i=elm\_max(M^1_i,M^2_i)$ for all $inv^3_i(l_i,X_i,M^3_i)\in I^3$
    
\end{itemize}
where $elm\_min()$ and $elm\_max()$ calculate element-wise minimum and maximum of vectors.

\begin{theorem}
$\mathscr{L}(\mathscr{A}_1) \cup \mathscr{L}(\mathscr{A}_2) \subseteq \mathscr{L}(\mathscr{A}_3)$
\end{theorem}
\textbf{Proof: }If we define $\Delta_G=\mathscr{A}_1.G.N-elm\_min(\mathscr{A}_1.G.N,\mathscr{A}_2.G.N)$ and $\Delta_I=elm\_max(\mathscr{A}_1.I.M,\mathscr{A}_2.I.M)-\mathscr{A}_1.I.M$, we have $\mathscr{A}_3=\mathcal{R}_1(\mathscr{A}_1,\Delta_G,\Delta_I)$.
Then according to Theorem 1, $\mathscr{L}(\mathscr{A}_1)\subset \mathscr{L}(\mathscr{A}_3)$.
Similarly we have $\mathscr{L}(\mathscr{A}_2)\subset \mathscr{L}(\mathscr{A}_3)$, therefore the theorem holds.

\subsubsection{$\mathcal{R}_3$: Remove Internal Receiving Actions}
\textbf{Intuition: }Edges with receiving actions can be taken only when the action is sent.
Therefore receiving actions are equivalent to guards on edges.
If receiving actions are removed and the action is not observable to the system, it is equivalent to setting the guard to True, or setting the $[N,M]$ interval to $[0,\infty]$, therefore increase behavior coverage.\\
\textbf{Prerequisite: }There exists $a\notin A^O$ and $a$ is a broadcast channel.
There is also no guard on transition $e=(l,a,r,l')$\\
\textbf{Rule: }
For a timed automaton $\mathscr{A}=\mathscr{A}_1|\mathscr{A}_2| \cdots \mathscr{A}_N$, if there exists\\ $\mathscr{A}_i=(L^i,l^i_0,X^i,A^i,E^i,G^i,I^i)$ and 
$\mathscr{A}_j=(L^j,l_0^j,X^j,A^j,E^j,G^j,I^j), i,j\in [1,N]$ such that $ a^i_m$ is a sending action, $a^j_n$ is a receiving action and $a^i_m,a^j_n\notin A^O$, create a new timed automaton $\mathscr{A}'=\mathcal{R}_3(\mathscr{A})=\mathscr{A}_1|\mathscr{A}_2\cdots\mathscr{A}_n$ such that $a^j_n=\emptyset$ for $\mathscr{A}_j$.

\begin{theorem}
$\mathscr{L}(\mathscr{A}')\subseteq \mathscr{L}(\mathcal{R}_3(\mathscr{A}))$ 
\end{theorem}
\textbf{Proof: } We first prove that $\mathcal{R}_3(\mathscr{A})$ is a timed simulation of $\mathscr{A}$. 
Since a receiving action is removed from a transition, that transition is always enabled.
Therefore when the sending action is taken, the new transition is enabled and can be taken at the same time, which satisfies the timed simulation requirement.
Timed simulation ensures $\mathscr{L}(\mathscr{A}')\subseteq \mathscr{L}(\mathcal{R}_3(\mathscr{A}))$.
We can then use prove by construction to show that $\mathscr{A}'$ has timed words that are not in $\mathscr{A}$, therefore the theorem holds.

    
In the next section, we use a simple case study to demonstrate the application of the proposed framework.

\begin{figure*}[t]
\centering
	\includegraphics[width=0.85\textwidth]{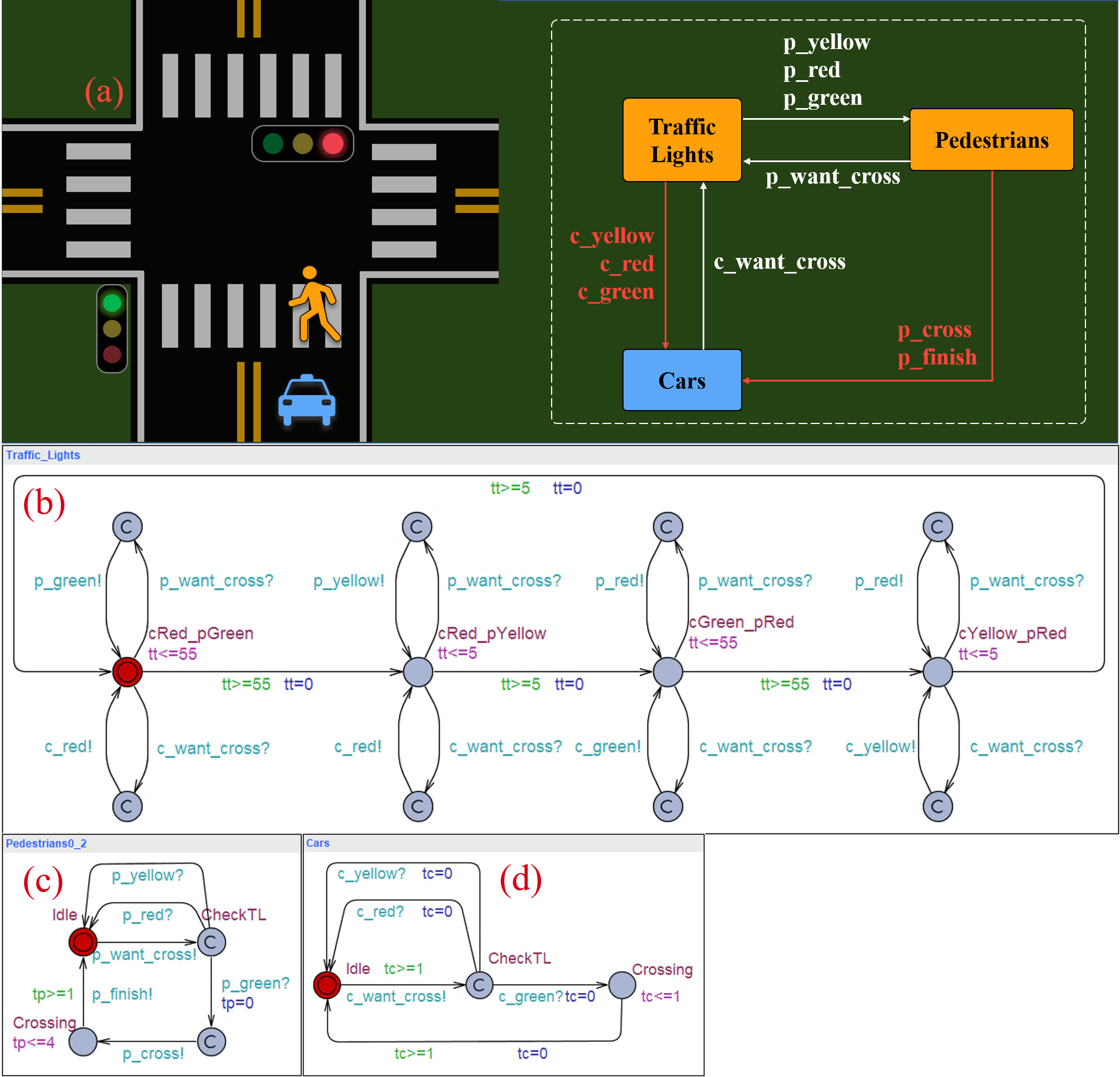}
	\caption{(a) The environment of the car includes a pedestrian and a traffic light. Actions in red represent observable actions. (b) The controller of the traffic lights sends the color of the light after receiving cross requests. (c) The model of a regular pedestrian who crosses only during green light within [1,4] seconds. (d) A simple and faulty car controller which crosses whenever the light is green.}
	\label{fig:ModelSchematic}
\end{figure*}

\section{Case Study: Ensuring Pedestrian Safety in Autonomous Driving}
Autonomous vehicles are CPSs which are required to safely operate within complex environment with large variabilities.
The environment consists of multiple agents with different states and parameters.
In this case study, we focus on a simple scenario in which an autonomous vehicle is crossing an intersection with traffic lights, with one pedestrian who may cross the road in front of the car (Fig. \ref{fig:ModelSchematic}.(a)).
The environment for the autonomous vehicle contains two components: the traffic lights and the pedestrian.
The autonomous vehicle can observe the color of the active light as well as the pedestrian's crossing and finishing actions.
The safety property is to prevent collision with the pedestrian, such that the car and the pedestrian cannot cross at the same time.

\subsection{Step 1: Base Environment Models}
Although traffic light is also part of the environment, due to its lack of variability, only the pedestrian model will be abstracted.
Domain experts can also decide to exclude certain components based on prior knowledge.
As shown in Fig. \ref{fig:absTree}, we start with two base pedestrian models: $Pedestrian0\_2$ who complies to traffic rules, and $Pedestrian0\_1$ who may cross the road when the traffic light is red.
\subsection{Step 2: Construction of the Abstraction Tree}
Although the base environment models already covers uncertain behaviors of the pedestrian (i.e. the timing of the cross intention and the duration of the cross action), the set of base environment models does not cover all possible observable behaviors of a pedestrian.
Abstraction rules were applied to the base models to construct the abstraction tree:\\
$\mathcal{R}_1$ was applied to both $Pedestrian0\_1$ and $Pedestrian0\_2$, increasing the time range for crossing the road from $[1,4]$ to $[1,15]$ and $[0,10]$ respectively, resulting in abstract models $Pedestrian1\_1$ and $Pedestrian1\_2$.
$\mathcal{R}_3$ was then applied to $Pedestrian1\_2$, removing interactions between the traffic light and the pedestrian, resulting in $Pedestrian2\_1$.
$Pedestrian2\_1$ and $Pedestrian1\_1$ now have the same structure, and therefore can be merged by $\mathcal{R}_2$, resulting in $Pedestrian3\_1$.

In this example we use the abstraction tree with $Pedestrian3\_1$ as root, although the behavior coverage of $Pedestrian3\_1$ can still be improved by applying $\mathcal{R}_1$. 
\begin{figure*}[h]
\centering
	\includegraphics[width=1.0\textwidth]{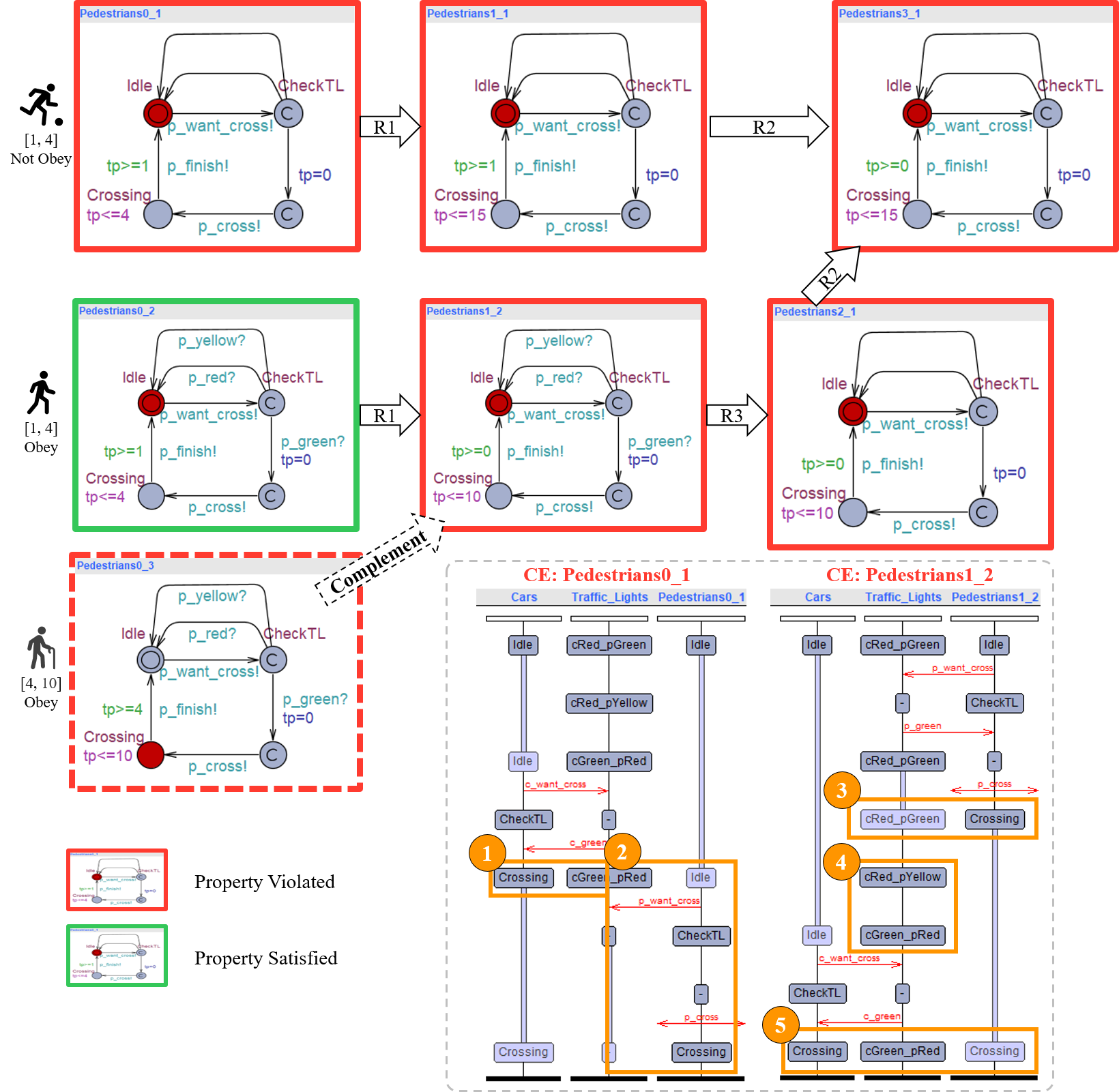}
	\caption{Abstraction tree of the environment (pedestrian) and two counter-examples returned from $Pedestrian0\_1$ and $Pedestrian1\_2$. }
	\label{fig:absTree}
\end{figure*}

\subsection{Step 3: Model Checking and Counter-example Refinement}
The abstraction tree of pedestrian models can then be used for closed-loop model checking of control algorithms of the autonomous vehicle.
In order to demonstrate the advantage of using the abstraction tree, we use a simple and faulty controller that crosses the road whenever the traffic light on its side is green (Fig. \ref{fig:ModelSchematic}.(d)). 
This safety property is specified using TCTL language as A[] not (P.Crossing and C.Crossing).
After traversing the abstraction tree, two refined counter-examples were returned which correspond to node $Pedestrian0\_1$ and $Pedestrian1\_2$ in the abstraction tree (Fig. \ref{fig:absTree}).

The counter-example from $Pedestrian0\_1$ is as expected since the pedestrian may cross the road when the traffic light is red ((2) in Fig. \ref{fig:absTree}), while the car is already crossing the road ((1) in Fig. \ref{fig:absTree}).

The counter-example from $Pedestrian1\_2$ shows a different mechanism.
Both the pedestrian and the car started to cross the road when the traffic light on their corresponding side was green ((3) and (5) in Fig. \ref{fig:absTree}).
The traffic light switched when the pedestrian was still crossing the road ((4) in Fig. \ref{fig:absTree}), triggering the crossing of the car and collision with the pedestrian.
\subsection{Step 4: Environment Model Refinement}
The safety property was violated in $Pedestrian1\_2$, but was satisfied in its child $Pedestrian0\_2$.
In order to pinpoint the environment condition in which the counter-example occurred, a new base model $Pedestrian0\_3$ can be obtained by "subtracting" $Pedestrian0\_2$ from $Pedestrian1\_2$ (Fig. \ref{fig:absTree}), and the property was also violated with the same counter-example mechanism. 
From $Pedestrian0\_3$ we can see that the collision happened due to the long crossing time of the pedestrian, which provided more interpretation to the counter-example.

\section{CONCLUSION}
Model checking of CPSs requires environment models that not only cover all possible environment conditions, but also provide interpretablity to the counter-examples.
Balancing these conflicting requirements requires expertise in both formal methods and the application domain, which prevents model checking from being effectively adopted for validation of CPSs.
In this project, a set of domain-independent abstraction rules for timed automata were developed to increase the coverage of environment models. 
A domain-independent framework for abstraction and refinement of environment models was proposed for model checking of CPSs.
The framework balances coverage and interpretability in environment models, and experts in the application domain can use model checking effectively without expertise in formal methods. 

Currently the base environment models are required to have the same model structure.
The next step is developing new abstraction rules that can remove locations, even components from environment models, so that this constraint can be relaxed.
The sequence for abstraction rule application may affect the completeness and abstraction level of counter-examples.
The next step is identifying the optimal sequence for abstraction rule application, and quantification of coverage.

\bibliographystyle{unsrt}
\bibliography{bibliography}

\begin{thebibliography}{10}

\bibitem{Mederror}
Vikram Jairam, Holly~M. Lincoln, Derek~W. Brown, Henry~S. Park, and Suzanne~B.
  Evans.
\newblock Error types and associations of clinically significant events within
  food and drug administration recalls of linear accelerators and related
  products.
\newblock {\em Practical Radiation Oncology}, 10(1):e8 -- e15, 2020.

\bibitem{cps1}
Magnus Lindahl, Paul Pettersson, and Wang Yi.
\newblock {Formal Design and Analysis of a Gear-Box Controller}.
\newblock In {\em Proc.\ of the 4{\em th} Workshop on Tools and Algorithms for
  the Construction and Analysis of Systems}, number 1384 in Lecture Notes in
  Computer Science, pages 281--297. Springer--Verlag, March 1998.

\bibitem{cps2}
F.~{Mercaldo}, F.~{Martinelli}, and A.~{Santone}.
\newblock Real-time scada attack detection by means of formal methods.
\newblock In {\em 2019 IEEE 28th International Conference on Enabling
  Technologies: Infrastructure for Collaborative Enterprises (WETICE)}, pages
  231--236, 2019.

\bibitem{STTT13}
Zhihao Jiang, Miroslav Pajic, Rajeev Alur, and Rahul Mangharam.
\newblock Closed-loop verification of medical devices with model abstraction
  and refinement.
\newblock {\em International Journal on Software Tools for Technology
  Transfer}, 16(2):191--213, Apr 2014.

\bibitem{Predicate}
T.~Yavuz.
\newblock Partial predicate abstraction and counter-example guided refinement.
\newblock {\em Journal of Logical and Algebraic Methods in Programming}, 110,
  2020.

\bibitem{timed_automata}
R.~Alur and D.~L. Dill.
\newblock {A Theory of Timed Automata}.
\newblock {\em Theoretical Computer Science}, 126:183--235, 1994.

\bibitem{uppaal}
Kim~G.\ Larsen, Paul Pettersson, and Wang Yi.
\newblock {{\sc Uppaal}\ in a Nutshell}.
\newblock {\em Int.\ Journal on Software Tools for Technology Transfer},
  1(1--2):134--152, October 1997.

\bibitem{cps3}
Alexandre David and Wang Yi.
\newblock Modelling and analysis of a commercial field bus protocol.
\newblock In {\em Proceedings of the 12th Euromicro Conference on Real Time
  Systems}, pages 165--172. IEEE Computer Society, 2000.

\bibitem{UPPAAL_Tutorial}
Gerd Behrmann, Alexandre David, and Kim~G Larsen.
\newblock A tutorial on uppaal 4.0.
\newblock {\em Department of computer science, Aalborg university}, 2006.

\bibitem{TheoryTA}
R.~Alur and D.~L. Dill.
\newblock A theory of timed automata.
\newblock {\em Theoretical Computer Science}, 126(2):183--235, 1994.

\end{thebibliography}

\appendix
\section{Formal proofs of abstractions rules}
First of all, we need to refer to the definitions of timed simulation and transition enabled interval.\\
\textbf{Definition: Timed simulation}\\
For two timed automata \(\mathscr{A}_1=(L^1,l^1_0,C^1,A^1,E^1,G^1,I^1)\) and\\
\(\mathscr{A}_2=(L^2,l^2_0,C^2,A^2,E^2,G^2,I^2)\), a timed
simulation relation~\cite{STTT13} is a binary relation
\(\operatorname{sim} \subseteq \Omega^{1} \times \Omega^{2}\) where
\(\Omega^{1}\) and \(\Omega^{2}\) are sets of states of
\(\mathscr{A}_1\) and \(\mathscr{A}_2\). We say \(\mathscr{A}_2\) time
simulates
\(\mathscr{A}_1\left(\mathscr{A}_1 \preceq_{t} \mathscr{A}_2\right)\) if
the following conditions holds:
\begin{enumerate}
\item
  initial states correspondence:
  \(\left(\left\langle l_{0}^{1}, \mathbf{0}\right),\left\langle l_{0}^{2}, \mathbf{0}\right\rangle\right) \in \operatorname{sim}\)
\item
  timed transition: For every
  \(\left(\left\langle l_{1}, v_{1}\right\rangle,\left\langle l_{2}, v_{2}\right\rangle\right) \in \operatorname{sim}\)\\
  if
  \(\left\langle l_{1}, v_{1}\right\rangle \stackrel{d}{\rightarrow}\left\langle l_{1}, v_{1}+d\right\rangle,\)
  there exists \(\left\langle l_{2}, v_{2}+d\right\rangle\)\\
  such that
  \(\left\langle l_{2}, v_{2}\right\rangle \stackrel{d}{\rightarrow}\left\langle l_{2}, v_{2}+d\right\rangle\)
  and
  \(\left(\left\langle l_{1},v_{1}+d\right\rangle,\right.\left.\left\langle l_{2}, v_{2}+d\right\rangle\right) \in \operatorname{sim}\)
\item
  discrete transition: for every
  \(\left(\left\langle l_{1}, v_{1}\right\rangle,\left\langle l_{2}, v_{2}\right\rangle\right) \in \operatorname{sim}\)\\
  1. if $a$ is an observable action,\\
  if
  \(\left\langle l_{1}, v_{1}\right\rangle \stackrel{a}{\rightarrow}\left\langle l_{1}^{\prime}, v_{1}^{\prime}\right\rangle,\)
  there exists
  \(\left\langle l_{2}^{\prime}, v_{2}^{\prime}\right\rangle\) such
  that\\
  \(\left\langle l_{2}, v_{2}\right\rangle \stackrel{a}{\rightarrow}\left\langle l_{2}^{\prime}, v_{2}^{\prime}\right\rangle\)
  and
  \(\left(\left\langle l_{1}^{\prime}, v_{1}^{\prime}\right\rangle,\left\langle l_{2}^{\prime}, v_{2}^{\prime}\right\rangle\right) \in \operatorname{sim}\)\\
  2. if $a$ is not an observable action,\\
  if
  \(\left\langle l_{1}, v_{1}\right\rangle \stackrel{a}{\rightarrow}\left\langle l_{1}^{\prime}, v_{1}^{\prime}\right\rangle,\)
  there exists
  \(\left\langle l_{2}^{\prime}, v_{2}^{\prime}\right\rangle\) such
  that\\
  \(\left\langle l_{2}, v_{2}\right\rangle \stackrel{a|\epsilon}{\rightarrow}\left\langle l_{2}^{\prime}, v_{2}^{\prime}\right\rangle\)
  and
  \(\left(\left\langle l_{1}^{\prime}, v_{1}^{\prime}\right\rangle,\left\langle l_{2}^{\prime}, v_{2}^{\prime}\right\rangle\right) \in \operatorname{sim}\)
  
\end{enumerate}
\textbf{Definition: Transition Enabled Interval}\\
Transition enabled interval is defined on an edge's guard and its output location's invariant, which indicates the time interval that an edge can be enabled. For an edge \(e = (l,a,r,l')\), the enabled interval is \(I(l)\wedge G(e)\).

For example, an edge \(e\) has the guard \(t\ge 3\) and \(l\) has the
invariant \(t\le 6\), the enabled interval is \([3,6]\).

If the enabled interval of an edge \(e\) is changed from \([a,b]\) to
\([a-\epsilon_g,b+\epsilon_i]\), with \(\epsilon_g \in \mathbb{N}\),
\(\epsilon_i \in \mathbb{N}\) and \(\epsilon_g+\epsilon_i > 0\), we say that the
enabled interval of the edge \(e\) is extended.

Note that as we assuming that there is no deadlock, therefore
\(I(l)\wedge G(e) \wedge I(l') = I(l)\wedge G(e)\).\\

Followings are formal proofs of the abstraction rules.
\begin{theorem}
$\mathscr{A}_2=\mathscr{R}_1(\mathscr{A}_1, \Delta_G, \Delta_I) \Rightarrow$ $\mathscr{A}_2$ timed simulates $\mathscr{A}_1$.
\end{theorem}

\begin{proof}
After applying the \(\mathscr{R}_1\) on \(\mathscr{A}_1\), we get \(\mathscr{A}_2=GI(\mathscr{A}_1, \Delta_G, \Delta_I)\) whose transition enabled intervals are extended. Because the only differences of \(\mathscr{A}_2\) from \(\mathscr{A}_1\) is the guards \(G\) and invariants \(I\).

Therefore, the proof idea is that for any timed transition or discrete transition, there exist transitions that have the same location \(l\) and clock assignment \(v\) and differ from the transition enabled interval.

Let \(\operatorname{sim} \subseteq \Omega^{1} \times \Omega^{2}\) where \(\Omega^{1}\) and \(\Omega^{2}\) are sets of states of \(\mathscr{A}_1\) and \(\mathscr{A}_2\). It can be seen that \(\mathscr{A}_2\) time simulates \(\mathscr{A}_1\left(\mathscr{A}_1 \preceq_{t} \mathscr{A}_2\right)\) because the following conditions holds:

\begin{enumerate}
\item
  initial states correspondence: \(\left(\left\langle l_{0}, \mathbf{0}\right),\left\langle l_{0}, \mathbf{0}\right\rangle\right) \in \operatorname{sim}\)
\item
  timed transition: \(\forall \left(\left\langle l_{1}, v_{1}\right\rangle,\left\langle l_{2}, v_{2}\right\rangle\right) \in \operatorname{sim}, where \left\langle l_{1}, v_{1}\right\rangle \in \mathscr{A}_1\) and \(\left\langle l_{2}, v_{2}\right\rangle \in \mathscr{A}_2\),

  we want to prove \(\left\langle l_{1}, v_{1}\right\rangle \stackrel{d}{\rightarrow}\left\langle l_{1}, v_{1}+d\right\rangle \Rightarrow \left\langle l_{2}, v_{2}\right\rangle \stackrel{d}{\rightarrow}\left\langle l_{2}, v_{2}+d\right\rangle\).\\
  If \(\left\langle l_{1}, v_{1}\right\rangle \stackrel{d}{\rightarrow}\left\langle l_{1}, v_{1}+d\right\rangle\), then we know that \(v_1\vDash inv_1(l_1,X_1,N_1)\).

  \begin{itemize}
  \item
    Because \(\mathscr{A}_2=GI(\mathscr{A}_1, \Delta_G, \Delta_I)\),
    then there exists \(\langle l_{2}, v_{2}\rangle=\langle l_{1}, v_{1}\rangle\) such
    that \(inv_2(l_2,X_2,N_2)=inv_1(l_1,X_1,N_1+\Delta_{I}[1])\).
  \item
    Because \(v_1\vDash inv_1(l_1,X_1,N_1)\), \(v_2=v_1\),
    then \(v_2\vDash inv_1(l_1,X_1,N_1+\Delta_{I}[1])=inv_2(l_2,X_2,N_2)\).
  \item
    Then we have \(\left\langle l_{2}, v_{2}\right\rangle \stackrel{d}{\rightarrow}\left\langle l_{2}, v_{2}+d\right\rangle\).
    Therefore \(\left(\left\langle l_{1},v_{1}+d\right\rangle,\right.\left.\left\langle l_{2}, v_{2}+d\right\rangle\right) \in \operatorname{sim}\).
  \end{itemize}
\item
  discrete transition: \(\forall \left(\left\langle l_{1}, v_{1}\right\rangle,\left\langle l_{2}, v_{2}\right\rangle\right) \in \operatorname{sim}, where \left\langle l_{1}, v_{1}\right\rangle \in \mathscr{A}_1\) and \(\left\langle l_{2}, v_{2}\right\rangle \in \mathscr{A}_2\).

  We want to prove \(\langle l_{1}, v_{1}\rangle \stackrel{a}{\rightarrow}\langle l_{1}^{\prime}, v_{1}^{\prime}\rangle \Rightarrow \exists \langle l_{2}^{\prime}, v_{2}^{\prime}\rangle \in \mathscr{A}_2, \langle l_{2}, v_{2}\rangle \stackrel{a}{\rightarrow}\langle l_{2}^{\prime}, v_{2}^{\prime}\rangle\).

  If \(\left\langle l_{1}, v_{1}\right\rangle \stackrel{a}{\rightarrow}\left\langle l_{1}^{\prime}, v_{1}^{\prime}\right\rangle\), we know that \(v_1\vDash g_1(e_1,X_1,N_1)\), where \(e_1=(l_1,l'_1)\).

  \begin{itemize}
  \item
    Because \(\mathscr{A}_2=GI(\mathscr{A}_1, \Delta_G, \Delta_I)\), then there exists \(\langle l'_{2}, v'_{2}\rangle=\langle l'_{1}, v'_{1}\rangle\) such that \(g_2(e_2,X_2,N_2)=g_1(e_1,X_1,N_1-\Delta_G[1])\).
  \item
    Because \(v_1\vDash g_1(e_1,X_1,N_1)\), \(v_2=v_1\), then \(v_2\vDash g_1(e_1,X_1,N_1-\Delta_G[1])=g_2(e_2,X_2,N_2)\), where \(e_2=(l_2,l'_2)\).
  \item
    Then we have \(\left\langle l_{2}, v_{2}\right\rangle \stackrel{a}{\rightarrow}\left\langle l_{2}^{\prime}, v_{2}^{\prime}\right\rangle\).
  \end{itemize}
\end{enumerate}

Therefore, \(\left(\left\langle l_{1}^{\prime}, v_{1}^{\prime}\right\rangle,\left\langle l_{2}^{\prime}, v_{2}^{\prime}\right\rangle\right) \in \operatorname{sim}\).
\end{proof}





\begin{theorem}
$\mathscr{A}_2$ timed simulates $\mathscr{A}_1$ $\Rightarrow$ $\mathscr{L}(\mathscr{A}_1) \subseteq \mathscr{L}(\mathscr{A}_2)$.
\end{theorem}

\begin{proof}
The language of a model is defined on the observable sending actions. 
Let $w$ be any timed word of $\mathscr{A}_1$ then there must exist a simulation procedure \\
$p_1= \left\langle l^1_0,\mathbf{0} \right\rangle  \left\langle l^1_1,v^1_1 \right\rangle  \left\langle l^1_2,v^1_2 \right\rangle ... \left\langle l^1_n,v^1_n \right\rangle $ of $\mathscr{A}_1$ that produces $w$.

Let $\Omega^1$ be the sets of states of $\mathscr{A}_1$ and $\Omega^2$ be the sets of states of $\mathscr{A}_2$. Let $sim\subseteq \Omega^1 \times \Omega^2$ be the timed simulation relation. Because $\mathscr{A}_2$ timed simulates $\mathscr{A}_1$, then we have\\
1. $\exists  \left\langle l^2_0,\mathbf{0} \right\rangle \in \Omega_2$ such that  $( \left\langle l^1_0,\mathbf{0} \right\rangle , \left\langle l^2_0,\mathbf{0} \right\rangle )\in sim$. We can set a global time $t^1_g=0$ for $\mathscr{A}_1$ and $t^2_g=0$ for $\mathscr{A}_2$ at the initial state.\\
2. $\forall \left\langle l_i,v_i \right\rangle , \left\langle l_{i+1},v_{i+1} \right\rangle \in \Omega^1$, 

(Timed Transition) If the two neighbor states have a time increase by $d$, i.e. $l^1_{i+1}=l^1_{i}$, $v^1_{i+1}=v^1_{i}+d$ and $\left\langle l^1_{i}, v^1_{i}\right\rangle \stackrel{d}{\rightarrow}\left\langle l^1_{i+1}, v^1_{i+1}\right\rangle$, then $t^1_g$ increase $d$ between the two states.\\
From the timed simulation relation we know that $\exists \left\langle l^2_j,v^2_j \right\rangle, \left\langle l^2_{j+1},v^2_{j+1} \right\rangle \in \Omega_2$ such that $l^2_{j+1}=l^2_{j}$, $v^2_{j+1}=v^2_{j}+d$, $\left\langle l^2_{j}, v^2_{j}\right\rangle \stackrel{d}{\rightarrow}\left\langle l^2_{j+1}, v^2_{j+1}\right\rangle$, which means $t^2_g$ can increase the same time $d$.

(Discrete Transition with Observable Events) If the two neighbor states have a transition that does not have an observable action, in other words, $\left\langle l^1_{i}, v^1_{i}\right\rangle \stackrel{a}{\rightarrow}\left\langle l^1_{i+1}, v^1_{i+1}\right\rangle$, where $a$ is an observable action.\\
From the timed simulation relation we know that $\exists \left\langle l^2_j,v^2_j \right\rangle, \left\langle l^2_{j+1},v^2_{j+1} \right\rangle \in \Omega_2$ such that $\left\langle l^2_{j}, v^2_{j}\right\rangle \stackrel{a}{\rightarrow}\left\langle l^2_{j+1}, v^2_{j+1}\right\rangle$, which means $\mathscr{A}_2$ can send the same observable signal $a$ as $\mathscr{A}_1$.


Until here we know
1. $\mathscr{A}_1$ and $\mathscr{A}_2$ have the same initial global time as $zero$.

2. For any timed transition of $\mathscr{A}_1$ that increases $t^1_g$ by $d$, there is the same procedure in $\mathscr{A}_2$ such that $t^2_g$ increases the same time $d$.

3. For any transition with observable event in $\mathscr{A}_1$, there is the same procedure that produces the same observable event in $\mathscr{A}_2$.

For any timed word $w$ of $\mathscr{A}_1$ that send some observable action at some time, $\mathscr{A}_2$ can simulate the same timed word, i.e. send the same observable action at the same time. 

Therefore, $\mathscr{L}(\mathscr{A}_1) \subseteq \mathscr{L}(\mathscr{A}_2)$.
\end{proof}




\end{document}